\documentclass[11pt]{article}
\usepackage{geometry}                % See geometry.pdf to learn the layout options. There are lots.
\usepackage{amsmath, amsthm, amssymb}
\usepackage{graphicx}
\usepackage{amssymb}
\usepackage{epstopdf}
\title{Optimal bounds for periodic mixtures of ferromagnetic interactions}
\author{Andrea Braides\\ \\Dipartimento di Matematica, Universit\`a di Roma Tor Vergata
\\ via della ricerca scientifica 1, 00133 Roma, Italy}
\date{}                                           % Activate to display a given date or no date

\def\Z{\mathbb Z}
\def\N{\mathbb N}
\def\R{\mathbb R}
\def\e{\varepsilon}

\begin{document}
\maketitle
%\section{}
%\subsection{}
\newtheorem{lemma}{Lemma}
\newtheorem{corollary}[lemma]{Corollary}
\newtheorem{proposition}[lemma]{Proposition}
\newtheorem{remark}[lemma]{Remark}
\newtheorem{example}[lemma]{Example}
\newtheorem{theorem}[lemma]{Theorem}
\newtheorem{definition}[lemma]{Definition}

\section{Introduction}
In this paper we give optimal bounds for the homogenization of periodic Ising systems of the
form
$$\sum_{ij}c_{ij}(u_i-u_j)^2
$$
where $u_i\in \{-1,1\}$, the sum runs over all nearest neighbours in a square lattice,
and the bonds $c_{ij}$ may take two positive values $\alpha$ and $\beta$ with 
$$
\alpha<\beta.
$$
Such bounds are given in terms of the {\em volume fraction} (proportion) $\theta$ of $\beta$-bonds
as follows.
To each such system we associate a {\em homogenized surface tension} $\varphi$.
We show that all possible such $\varphi$ are the (positively homogeneous of degree one) convex functions
such that
\begin{equation}\label{a}
\alpha(|\nu_1|+|\nu_2|)\le \varphi(\nu)\le c_1|\nu_1|+c_2|\nu_2|,
\end{equation}
where the coefficients $c_1$ and $c_2$ satisfy
\begin{equation}\label{b}
c_1\le \beta,\quad c_2\le \beta,\qquad c_1+c_2= 2(\theta\beta+(1-\theta)\alpha).
\end{equation}

The continuous counterpart of this problem is the determination of optimal bounds
for Finsler metrics  obtained from the homogenization of periodic Riemannian
metrics (see \cite{AcBu,BDF,BBF}) of the form
$$
\int_a^b a(u(t))|u'|^2dt,
$$
and $a(u)$ is a periodic function in $\R^2$ taking only the values $\alpha$ and $\beta$.
This problem has been studied in \cite{DP}, where it is shown that homogenized metrics
satisfy
$$
\alpha\le \varphi(\nu)\le (\theta\beta+(1-\theta)\alpha),
$$
but the optimality of such bounds is not proved.
A `dual' equivalent formulation in dimension $2$ is obtained by considering the homogenization of 
periodic perimeter functionals of the form
$$
\int_{\partial A}a(x) \,d{\cal H}^1(x)
$$
with the same type of $a$ as above (see \cite{AB,AI}). The corresponding $\varphi$ in this case 
can be interpreted as the homogenized surface tension of the homogenized perimeter functional.

The discrete setting allows to give a (relatively) easy description of the optimal bounds in
a way similar to the treatment of mixtures of linearly elastic discrete structures \cite{BF}.
The bounds obtained by sections and by averages in the elastic case have as counterpart 
{\em bounds by projection}, where the homogenized surface tension is estimated from below by
considering the minimal value of the coefficient on each section, and {\em bounds by averaging},
where coefficients on a section are substituted with their average. The discrete setting allows to
construct (almost-)optimal periodic geometries, which optimize one type or the other of bound in each direction. We shortly describe the `extreme' geometries in Fig.~\ref{series} and Fig.~\ref{parallel}, where $\alpha$-connections are represented as dotted lines, $\beta$-connections are represented as solid lines, and the nodes with the value $+1$ or $-1$ as white circles or black circles, respectively.  In Fig.~\ref{series} there are pictured the periodicity cell of a mixture giving as a result the lower bound $\alpha(|\nu_1|+|\nu_2|)$ and an interface with minimal energy.  
\begin{figure}[h]
\centerline{\includegraphics[width=.25\textwidth]{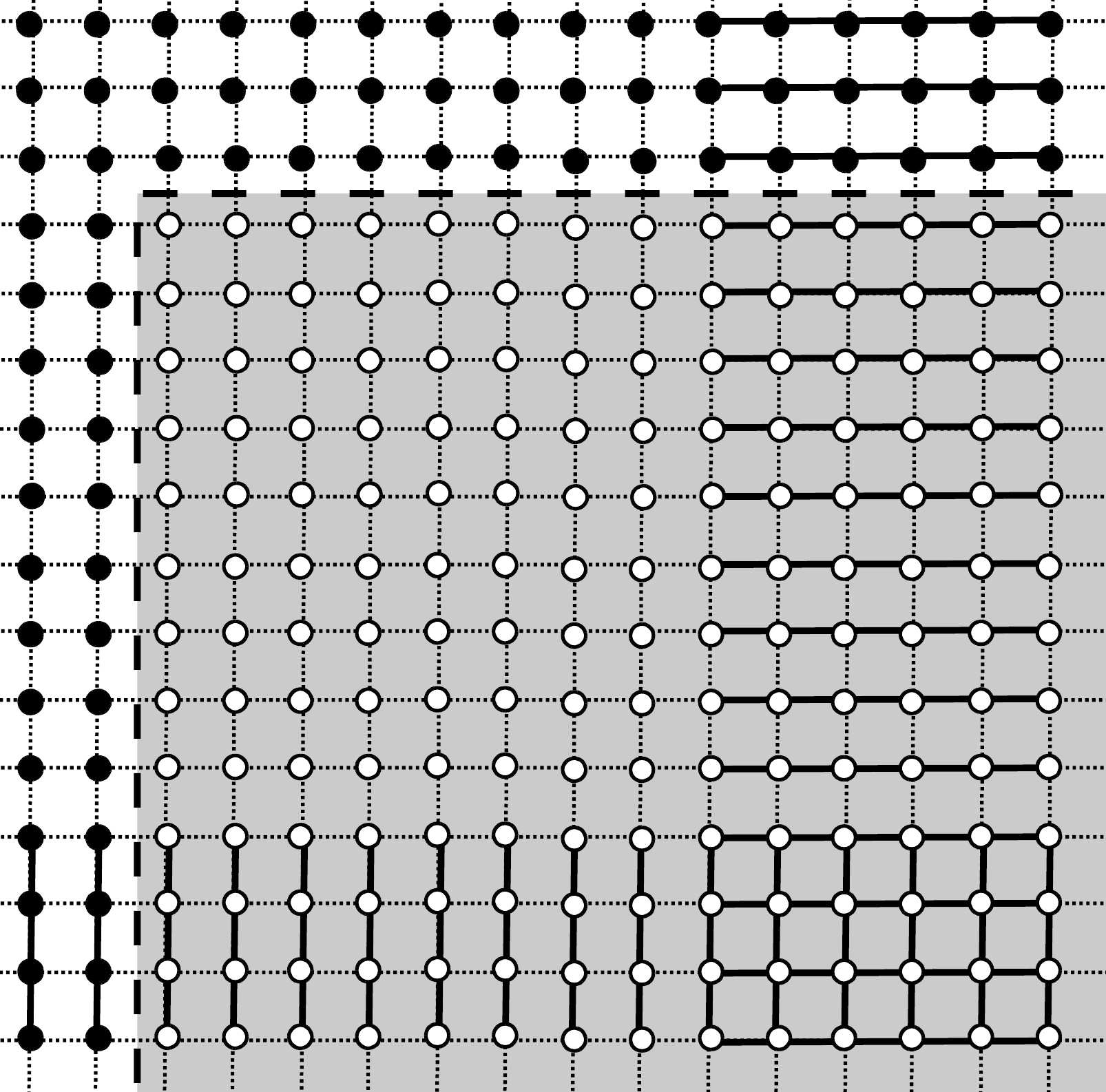}}
 \caption{periodicity cell for a mixture giving the lower bound}\label{series}
 \end{figure}
 Fig.~\ref{parallel} represents the periodicity cell of a mixture giving a upper bound of the form $c_1|\nu_1|+c_2|\nu_2|$. Note that the interface pictured in that figure crosses exactly
a number of bonds proportional to the percentage $\theta_v$ of $\beta$-bonds in the horizontal direction.
 \begin{figure}[h]
\centerline{\includegraphics[width=.25\textwidth]{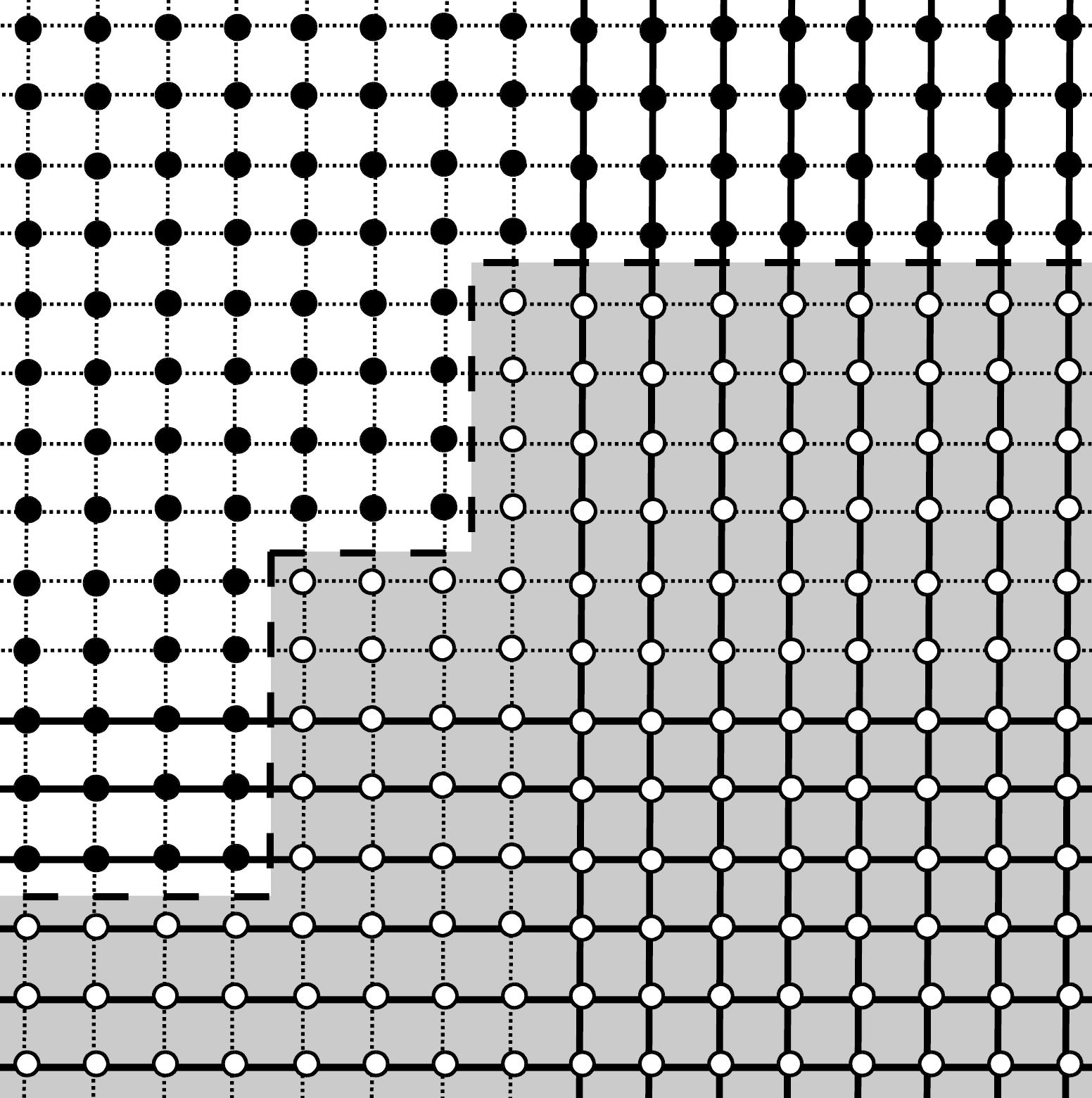}}
 \caption{periodicity cell for a mixture giving an upper bound}\label{parallel}
 \end{figure}

It must be noted that, contrary to the elastic case, the bounds (i.e., the sets of possible $\varphi$) are increasing with $\theta$, and in particular they always contain the minimal surface tension $\alpha(|\nu_1|+|\nu_2|)$. 

We can picture the bounds in terms of their Wulff shape; i.e.,
the solutions $A_\varphi$ centered in $0$ to the problem
$$
\max\Bigl\{ |A|: \int_{\partial A} \varphi(\nu)d{\cal H}^1(x)=1\Bigr\}.
$$
If $\varphi(\nu)= c_1|\nu_1|+c_2|\nu_2|$ then such Wulff shape is the rectangle centered in $0$ with one vertex in $(1/(8c_2),1/(8c_1))$. A general $\varphi$ satisfying (\ref{a}) and (\ref{b}) corresponds to a convex symmetric set contained in the square of side length
$1/(4\alpha)$ (that is, the Wulff shape corresponding to $\alpha(|\nu_1|+|\nu_2|)$) and containing one of such rectangles for $c_1$ and $c_2$ satisfying (\ref{b}). The envelope of the vertices of such rectangles lies in the curve
\begin{equation}\label{bondeq}
{1\over |x_1|}+{1\over |x_2|}= 16(\theta\beta+(1-\theta)\alpha)
\end{equation}
(see Fig.~\ref{bound1}).
\begin{figure}[h]
\centerline{\includegraphics[width=.50\textwidth]{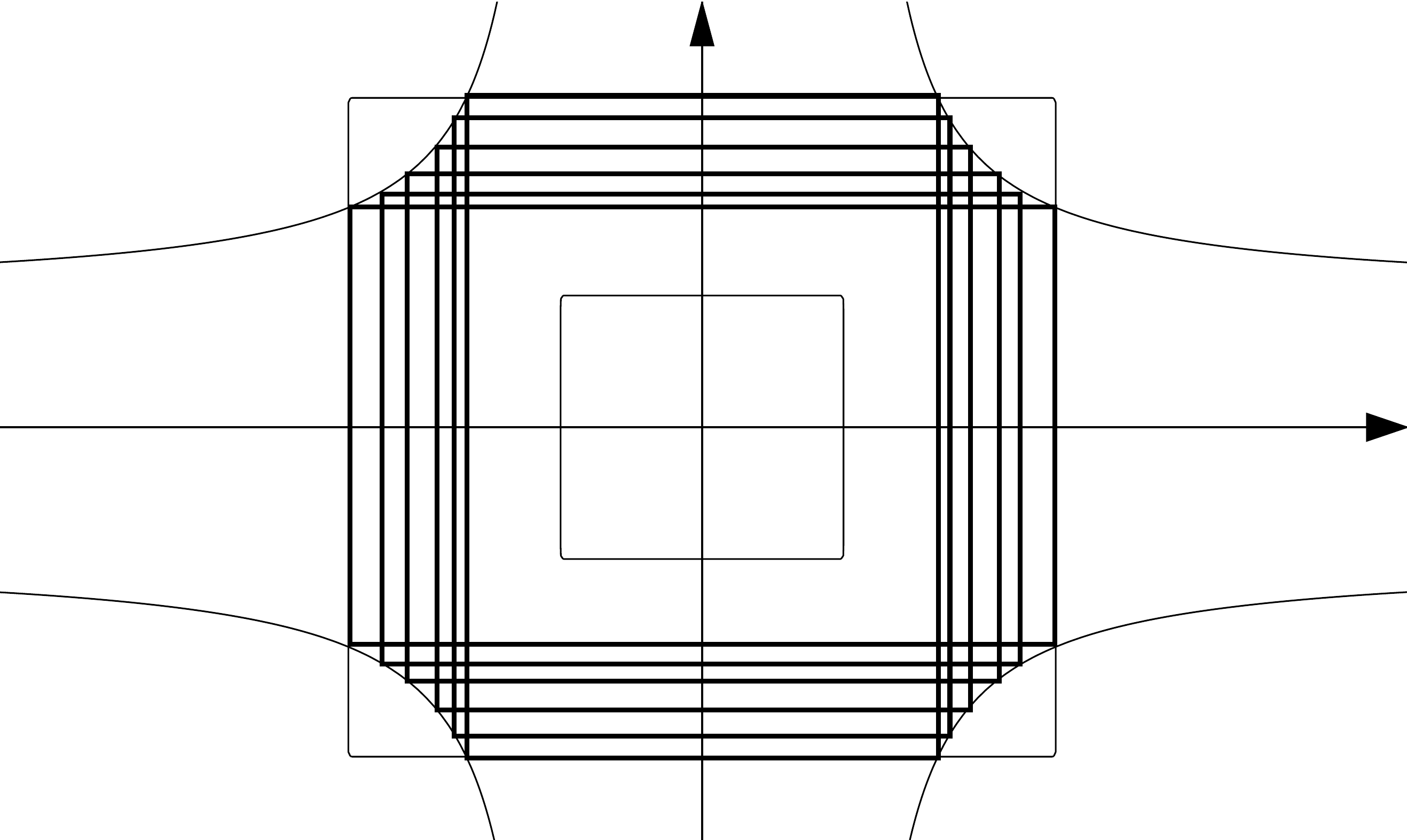}}
 \caption{envelope of rectangular Wulff shapes}\label{bound1}
 \end{figure}
In terms of such envelope, we can describe the Wulff shapes of $\varphi$ as follows:
if $\theta\le1/2$ then it is any symmetric convex set contained in the square of side length
$1/(4\alpha)$ and intersecting the four portions of the set of points satisfying (\ref{bondeq})
contained in that square (see Fig.~\ref{bound2}(a));
if $\theta\ge1/2$ then it is any symmetric convex set contained in the square of side length
$1/(4\alpha)$ and intersecting the four portions of the set of points satisfying (\ref{bondeq})
with $|x_1|\ge 1/(8\beta)$ and $|x_2|\ge 1/(8\beta)$ contained in that square (see Fig.~\ref{bound2}(b)). This second condition is automatically satisfied if $\theta\le 1/2$.
\begin{figure}[h]
\centerline{\includegraphics[width=.30\textwidth]{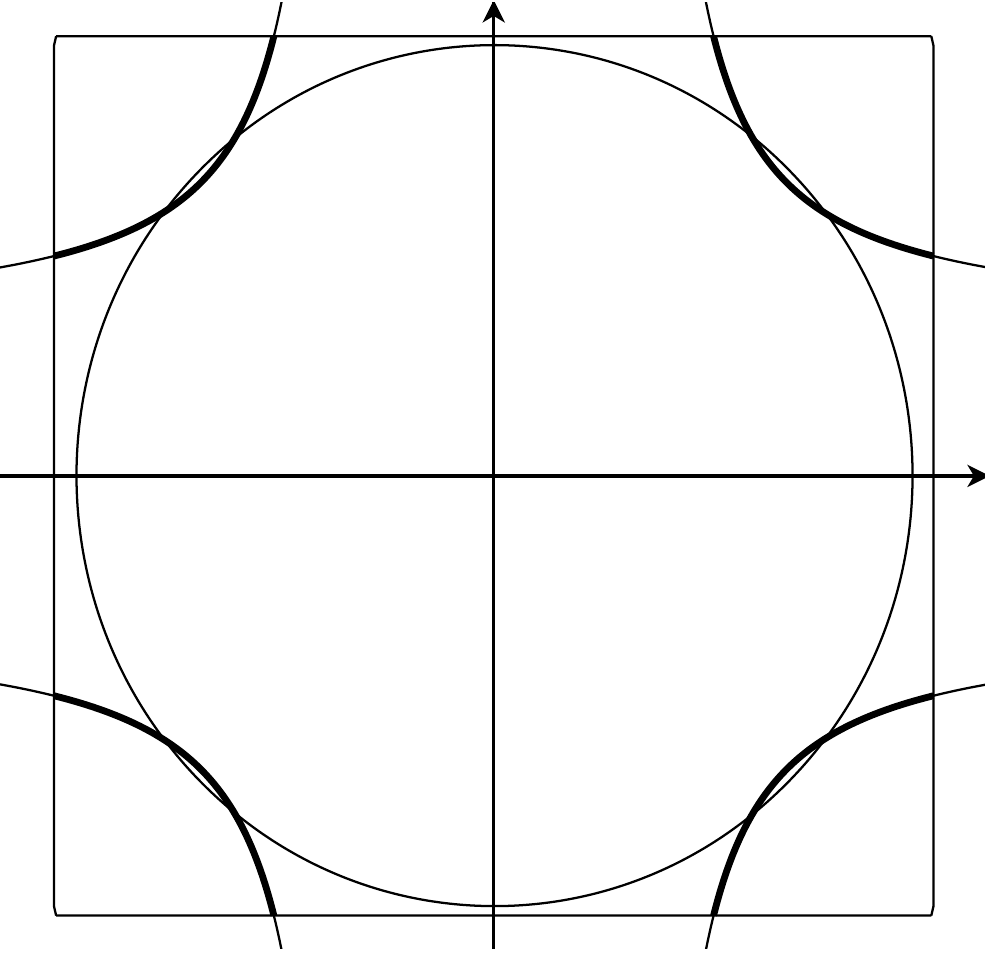}\qquad\qquad\qquad \includegraphics[width=.30\textwidth]{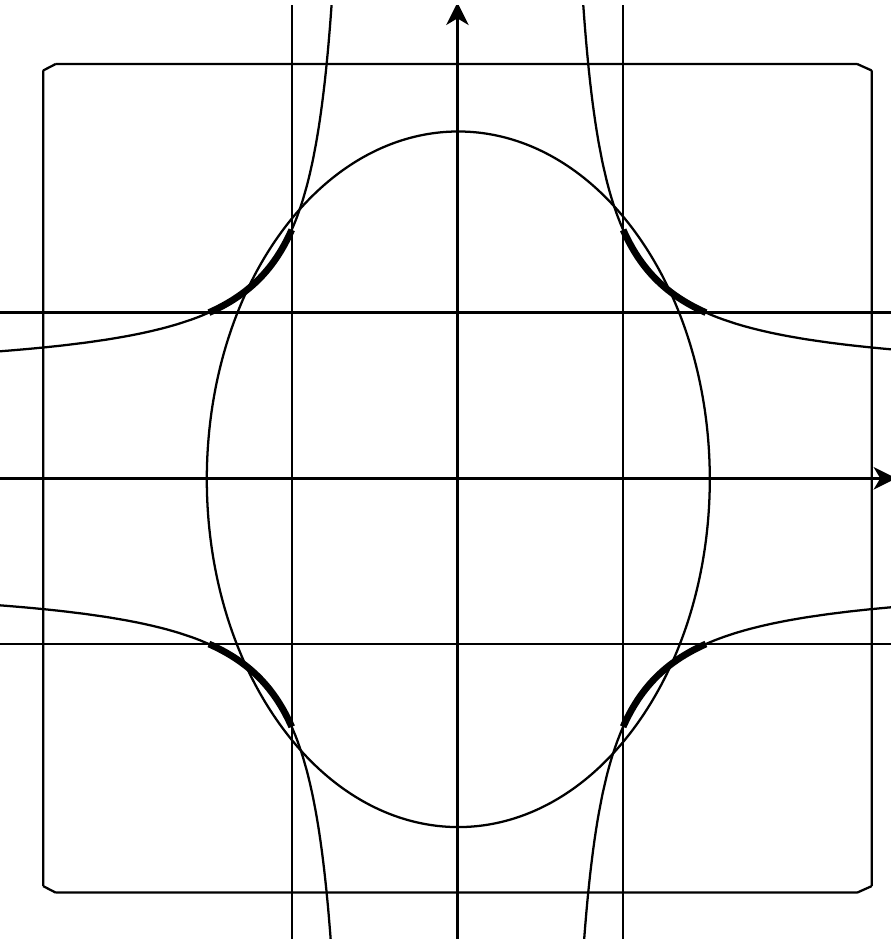}}
 \caption{possible Wulff shapes with: (a) $\theta\le 1/2$ and  (b) $\theta\ge 1/2$}\label{bound2}
 \end{figure}
 
\section{Setting of the problem}
We consider a discrete system of nearest-neighbour interactions in dimension two
with coefficients $c_{ij}=c_{ji}\ge 0$, $i,j\in\Z^2$. The corresponding ferromagnetic spin energy is
\begin{equation}
F(u)=\sum_{ij} c_{ij} (u_i-u_j)^2,
\end{equation}
where $u:\Z^2\to\{-1,1\}$, $u_i=u(i)$, and the sum runs over 
the set of {\em nearest neighbours} or {\em bonds} in $\Z^2$, which is denoted by
$$
{\cal Z}=\{(i,j)\in \Z^2\times \Z^2: |i-j|=1\}.
$$
Such energies correspond to inhomogeneous surface energies on the continuum \cite{ABC,BP}.

\begin{definition} Let $\{c_{ij}\}$ be indices as above with $\inf_{ij}c_{ij}>0$ and periodic; i.e., 
such that there exists $T\in \N$ such that 
$$
c_{(i+T)\,j}= c_{i\,(j+T)}= c_{ij}.
$$
Then, we define the {\em homogenized energy density} of $\{c_{ij}\}$ as the convex positively homogeneous function of degree one $\varphi:\R^2\to[0,+\infty)$ such that for all $\nu\in S^1$ we have
\begin{equation}\label{phi}
\varphi(\nu)=\lim_{R\to+\infty}\inf\Bigr\{ {1\over R}\sum_{n=1}^{N} c_{i_{n}j_{n}}: i_N-i_0=\nu^\perp R+o(R)\Bigr\}.
\end{equation}
The infimum is taken over all {\em paths of bonds}; i.e., pairs $(i_n, j_n)$ such that the unit segment
centred in ${i_n+j_n\over 2}$ and orthogonal to $i_n-j_n$ has an endpoint in common with  the unit segment centred in ${i_{n-1}+j_{n-1}\over 2}$ and orthogonal to $i_{n-1}-j_{n-1}$.
This is a good definition thanks to \rm\cite{BP}. 
\end{definition}

\begin{remark}\label{gamma}
\rm The definition above can be interpreted in terms of a passage from a discrete to a continuous description as follows. We consider the scaled energies
$$
E_\e(u)={1\over 8} \sum_{ij}\e c^\e_{ij}(u_i-u_j)^2,
$$
where $u:\e\Z^2\to\{-1,1\}$, the factor $1/8$ is a normalization factor, the sum runs over nearest neighbours in $\e\Z^2$, and
$$
c^\e_{ij}= c_{{i\over\e},{j\over\e}}.
$$
Upon identifying $u$ with its piecewise-constant interpolation, we can regard these energies as defined 
on $L^1(\R^2)$. Their $\Gamma$-limit in that space is infinite outside $BV_{\rm loc}(\R^2,\{\pm1\})$, 
where it has the form
$$
F_{\varphi}(u)=\int_{\partial\{u=1\}}\varphi(\nu)d{\cal H}^1
$$
with $\varphi$ as above.
\end{remark}

\bigskip\noindent
{\bf Periodic mixtures of two types of bonds.} We will consider the case when
\begin{equation}\label{abe}
\hbox{ $c_{ij}\in\{\alpha, \beta\}$ with $0<\alpha<\beta$;}
\end{equation}
If we have such coefficients, we define the {\em volume fraction} of $\beta$-bonds as 
\begin{equation}
\theta(\{c_{ij}\})={1\over 4T^2}\#\Bigl\{(i,j)\in {\cal Z}: {i+i\over 2}\in [0,T)^2,\  c_{ij}=\beta\Bigr\}.
\end{equation}

\begin{definition}\label{dephi}
Let $\theta\in [0,1]$. The set of {\em homogenized energy densities
of mixtures of $\alpha$ and $\beta$ bonds, with volume fraction $\theta$} (of $\beta$ bonds) 
is defined as
\begin{eqnarray}
{\bf H}_{\alpha,\beta}(\theta)&=&\nonumber
\Bigl\{ \varphi:\R^2\to [0,+\infty): \hbox{ there exist } \theta_k\to \theta,\  \varphi_k\to\varphi \hbox{ and } \{c^k_{ij}\} \\
&&
\hbox{ with }\theta(\{c^k_{ij}\})=\theta_k \hbox{ and }\varphi_k\hbox { homogenized energy density of }\{c^k_{ij}\}\Bigr\}.
\end{eqnarray}
\end{definition}

The following theorem completely characterizes the set ${\bf H}_{\alpha,\beta}(\theta)$.

\begin{theorem}[optimal bounds]\label{main}
The elements of the set ${\bf H}_{\alpha,\beta}(\theta)$ are all even convex positively homogeneous functions of degree one $\varphi:\R^2\to[0,+\infty)$ such that
\begin{equation}
\alpha(|x_1|+|x_2|)\le \varphi(x_1,x_2)\le c_1|x_1|+ c_2|x_2|
\end{equation}
for some $c_1, c_2\le\beta$ such that
\begin{equation}
c_1+c_2= 2(\theta\beta+(1-\theta)\alpha).
\end{equation}
\end{theorem}

Note that the lower bound for functions in ${\bf H}_{\alpha,\beta}(\theta)$ is independent of $\beta$.
Note moreover that in the case $\theta=1$ we have all functions satisfying the trivial bound
\begin{equation}\label{tri}
\alpha(|x_1|+|x_2|)\le \varphi(x_1,x_2)\le \beta (|x_1|+|x_2|).
\end{equation}
This is due to the fact that in that case by considering $\theta_k\to 1$ we allow a vanishing 
volume fraction of $\alpha$ bonds, which is nevertheless sufficient to allow for all possible $\varphi$.

\section{Optimality of bounds}

We first give two bounds valid for every set of periodic coefficients $\{c_{ij}\}$. 

\begin{proposition}[bounds by projection]\label{bpp}
Let $\varphi$ be the homogenized energy density of $\{c_{ij}\}$; then we have
\begin{equation}\label{p-bound}
\varphi(x)\ge c^p_1|x_1|+c^p_2|x_2|,
\end{equation}
where 
\begin{equation}
c^p_1= {1\over T} \sum_{k=1}^T\min\{c_{ij}:  i_2=j_2=k\}
\end{equation}
and 
\begin{equation}
c^p_2= {1\over T} \sum_{k=1}^T\min\{c_{ij}:  i_1=j_1=k\}.
\end{equation}
\end{proposition}

\begin{proof}
The lower bound (\ref{p-bound}) immediately follows from the definition of $\varphi$, by subdividing 
the contributions of $c_{i_{n-1} i_n}$ in (\ref{phi}) into those with  $(i_n)_2=(i_{n-1})_2$ (or equivalently 
such that $i_n-i_{n-1}=\pm e_1$) and 
those with $(i_n)_1=(i_{n-1})_1$ (or equivalently  $i_n-i_{n-1}=\pm e_2$, and estimating
$$
c_{i_{n-1} i_n}\ge \min\{ c_{ij}: i_2=j_2=(i_n)_2\}
$$
and 
$$
c_{i_{n-1} i_n}\ge \min\{ c_{ij}: i_1=j_1=(i_n)_1\},
$$
respectively, in the two cases.
\end{proof}

\begin{proposition}[bounds by averaging]\label{bba}
Let $\varphi$ be the homogenized energy density of $\{c_{ij}\}$; then we have
\begin{equation}\label{p-bound2}
\varphi(x)\le c^a_1|x_1|+c^a_2|x_2|,
\end{equation}
where $c^a_1$ is the average over {\em horizontal bonds}
\begin{equation}
c^a_1= {1\over T^2} \sum\Bigl\{c_{ij}:  {i+j\over 2}\in [0,T)^2, i_2=j_2\Bigr\}
\end{equation}
and $c^a_2$ is the average over {\em vertical bonds}
\begin{equation}
c^a_2= {1\over T^2} \sum\Bigl\{c_{ij}:  {i+j\over 2}\in [0,T)^2, i_1=j_1\Bigr\}.
\end{equation}
\end{proposition}

\begin{proof}
The proof is obtained by construction of a suitable competitor $\{i_n,j_n\}$ for the characterization 
(\ref{phi}) of $\varphi$. To that end let $n_1, n_2\in\{1,\ldots, T\}$ be such that 
$$
{1\over T}\sum_{k=1}^T c_{(n_1-1,k), (n_1, k)}\le
{1\over T^2} \sum\Bigl\{c_{ij}:  {i+j\over 2}\in [0,T)^2, i_2=j_2\Bigr\}
$$
and 
$$
{1\over T}\sum_{k=1}^T c_{(k,n_2-1), (k,n_2)}\le
{1\over T^2} \sum\Bigl\{c_{ij}:  {i+j\over 2}\in [0,T)^2, i_1=j_1\Bigr\}.
$$
Up to a translation, we may suppose that $n_1=n_2=1$.
It is not restrictive to suppose that $\nu_1\ge0$ and $\nu_2\ge 0$.
We define $i_0= (\lfloor R\nu_2\rfloor,0)$
and $i_N=(0,\lfloor R\nu_1\rfloor)$. It suffices then to take in Definition \ref{dephi} 
the path of bonds $\{i_n,j_n\}$ obtained by concatenating the two paths of bonds defined by
$$
i^1_n=(\lfloor R\nu_2\rfloor-n, 0), \qquad j^1_n=(\lfloor R\nu_2\rfloor-n, 1),\qquad n=0, \ldots, \lfloor R\nu_2\rfloor-1
$$
and
$$
i^2_n=(0,n), \qquad j^2_n=(1, n),\qquad n=1, \ldots, \lfloor R\nu_1\rfloor\,.
$$
We then have 
\begin{eqnarray*}
\lim_{R\to+\infty}{1\over R}\Bigl(\sum_{n=1}^{\lfloor R\nu_2\rfloor} c_{(n,0)\,(n,1)}+
\sum_{n=1}^{\lfloor R\nu_1\rfloor} c_{(0,n)\,(1,n)}\Bigr)=|\nu_2|{1\over T} \sum_{n=1}^T c_{(n,0)\,(n,1)}+
|\nu_1|{1\over T} \sum_{n=1}^T c_{(0,n)\,(1,n)},
\end{eqnarray*}
and the desired inequality.
\end{proof}

We now specialize the previous bound to mixtures of two types of bonds.
Given $\{c_{ij}\}$ satisfying {\rm(\ref{abe})} we define the {\em volume fraction of horizontal $\beta$-bonds} as 
\begin{equation}\label{th}
\theta_h(\{c_{ij}\})={1\over T^2}\#\Bigl\{(i,j)\in {\cal Z}: {i+i\over 2}\in [0,T)^2,\  c_{ij}=\beta, i_2=j_2\Bigr\}.
\end{equation}
and the {\em volume fraction of vertical $\beta$-bonds} as 
\begin{equation}\label{tv}
\theta_v(\{c_{ij}\})={1\over T^2}\#\Bigl\{(i,j)\in {\cal Z}: {i+i\over 2}\in [0,T)^2,\  c_{ij}=\beta, i_1=j_1\Bigr\}.
\end{equation}
Note that
\begin{equation}\label{thh}
{\theta_h(\{c_{ij}\})+\theta_v(\{c_{ij}\})\over 2}= \theta(\{c_{ij}\}).
\end{equation}

\begin{proposition} Let $\{c_{ij}\}$ satisfy {\rm(\ref{abe})}, let $\theta_h=\theta_h(\{c_{ij}\})$
and $\theta_v=\theta_v(\{c_{ij}\})$, and let $\varphi$ be the homogenized energy density of $\{c_{ij}\}$. 
Then 
\begin{equation}
\varphi(\nu)\le (\theta_h\beta+(1-\theta_h)\alpha)|\nu_1|+ (\theta_v\beta+(1-\theta_v)\alpha)|\nu_2|
\end{equation}
\end{proposition}

\begin{proof}
It suffices to rewrite $c^a_1$ and $c^a_2$ given by the previous proposition using (\ref{th}) and (\ref{tv}).
\end{proof}

The previous proposition, together with (\ref{thh}) and the trivial bound (\ref{tri})
gives the bounds in the statement of Theorem \ref{main}. We now prove their optimality.
First we deal with a special case, from which the general result will be deduced by
approximation.

\begin{proposition}\label{special}
Let $$\varphi(\nu)= c_1|\nu_1|+c_2|\nu_2|$$
with $\alpha\le c_1, c_2\le \beta$ and 
\begin{equation}
c_1+c_2\le 2(\beta\theta+ (1-\theta) \alpha)
\end{equation} 
for some $\theta\in (0,1)$. Then $\varphi\in {\bf H}_{\alpha,\beta}(\theta)$.
\end{proposition}

\begin{proof}
The case $\theta=1$ is trivial. In the other cases, since the set of $(c_1,c_2)$ as above
coincides with the closure of its interior, by approximation it suffices to consider the case when 
indeed
\begin{equation}
\alpha<c_1,c_2<\beta,\qquad c_1+c_2< 2(\beta\theta+ (1-\theta) \alpha).
\end{equation}
In particular, we can find $\theta_1\in (0,1)$ and $\theta_2\in(0,1)$ such that 
$\theta_1+\theta_2=2\theta$ and 
\begin{equation}
c_1<\beta\theta_1+ (1-\theta_1) \alpha),\qquad c_2<  \beta\theta_2+ (1-\theta_2) \alpha .
\end{equation}
We then write 
\begin{equation}
c_1=\beta t_1+ (1-t_1) \alpha),\qquad c_2=  \beta t_2+ (1-t_2) \alpha .
\end{equation}
for some $t_1<\theta_1$ and  $t_2<\theta_2$.

We construct $\{c_{ij}\}$ with period $T\in\N$ and with
$$
\theta_h(\{c_{ij}\})=\theta_1, \qquad \theta_v(\{c_{ij}\})=\theta_2
$$
by defining separately the horizontal and vertical bonds.
Upon an approximation argument we may suppose that $N_i= t_i T\in\N$, and that $T^2\theta_i\in \N$ for $i=1,2$. We only describe the construction for the horizontal bonds. 
We define
$$
c_{(j,n), (j+1,n)}=\begin{cases} \beta  &\hbox{ if $j=1,\ldots, T$ and $n=1,\ldots N_1$}\cr
\alpha &\hbox{  if $j=0$ and $n=N_1+1,\ldots T$}
\end{cases}
$$
and any choice of $\alpha$ and $\beta$ for other indices $i,j$, only subject to the 
total constraint that $\theta_h(\{c_{ij}\}=\theta_1$.
With this choice of $c_{ij}$ we have 
$$
\min\{ c_{ij}: i_2=j_2=n\}=\begin{cases} \beta  &\hbox{ if $n=1,\ldots N_1$}\cr
\alpha &\hbox{  if  $n=N_1+1,\ldots T$}
\end{cases}
$$
The analogous construction for vertical bonds gives
$$
\min\{ c_{ij}: i_1=j_1=n\}=\begin{cases} \beta  &\hbox{ if $n=1,\ldots N_2$}\cr
\alpha &\hbox{  if  $n=N_2+1,\ldots T$}
\end{cases}
$$
Then, Proposition \ref{bpp} gives that the homogenized energy density of $\{c_{ij}\}$ satisfies
$$
\varphi(\nu)\ge c_1|\nu_1|+c_2|\nu_2|.
$$

To give a lower bound we use the same construction of the proof of Proposition \ref{bba}, after noticing that the vertical and horizontal paths with $i^1_n=(0,n)$, $j^1_n=(1,n)$ or  $i^2_n=(n,0)$, $j^2_n=(n,1)$ are such that
$$
{1\over T}\sum_{n=1}^T c_{i^1_n,j^1_n}= c_1,
\qquad
{1\over T}\sum_{n=1}^T c_{i^2_n,j^2_n}= c_2.
$$
In this way we obtain the estimate
$$
\varphi(\nu)\le c_1|\nu_1|+c_2|\nu_2|.
$$
and hence the desired equality.
\end{proof}

\begin{proof}[Proof of Theorem \rm\ref{main}]
In order to conclude the proof we will use the energy densities obtained in the previous proposition to approximate all $\varphi$ satisfying the bounds. In order to do this, we note that, thanks to Remark \ref{gamma} we may use the fact that the class of functionals with integrands in ${\bf H}_{\alpha,\beta}(\theta)$ is closed under $\Gamma$-convergence. Hence, it is not restrictive to make some simplifying hypotheses on the function $\varphi$.

We may suppose that 
$$
\alpha (|\nu_1|+|\nu_2|)< \varphi(\nu)< (\beta \theta_1+ (1-\theta_1)\alpha)|\nu_1|+
(\beta \theta_2+ (1-\theta_2)\alpha)|\nu_2|=: c_1|\nu_1|+
c_2|\nu_2|
$$
for some $\theta_1,\theta_2\in (0,1)$, and that

$\bullet$ the set $\{x: \varphi(x)\le 1\}$ is a convex symmetric polyhedron with vertices corresponding to integer directions $\pm\nu^1,\ldots,\pm\nu^N$; i.e. such that there exist $r^j\in\R$ such that $r^j\nu^j\in\Z^d$.

The surface energy related to such $\varphi$ can be obtained as a $\Gamma$-limit of energies of the form
$$
F_\e(u)=\int_{\partial\{u=1\}}f\Bigl({x\over\e},\nu\Bigr)d{\cal H}^1,
$$
where $f(\cdot, \nu)$ is $1$-periodic and has the form
$$
f(y,\nu)= \begin{cases}
\varphi(\nu^k) & \hbox{ if } y\in \{ t(\nu^k)^\perp: t\in \R\}+\Z^2, \ k=1,\ldots, N
\\
c_1|\nu_1|+
c_2|\nu_2| & \hbox{ otherwise}
\end{cases}
$$
This can be proved as in \cite{BBF} or \cite{AI}, whose construction, where we have 
$\beta$ in place of $c_1|\nu_1|+c_2|\nu_2|$, works also in this case. 

Note that we can rewrite the values
$$
\varphi(\nu^k)=  (\beta \theta^k_1+ (1-\theta^k_1)\alpha)|\nu^k_1|+
(\beta \theta^k_2+ (1-\theta^k_2)\alpha)|\nu^k_2|=: c^k_1|\nu^k_1|+
c^k_2|\nu^k_2|
$$
with $c^k_1, c^k_2$ satisfying 
$$
c^k_1+c^k_2\le c_1+c_2.
$$
We can therefore consider equivalently
$$
f(y,\nu)= \begin{cases}
c^k_1|\nu_1|+c^k_2|\nu_2| & 
\hbox{ if } y\in \{ t(\nu^k)^\perp: t\in \R\}+\Z^2, \ k=1,\ldots, N
\\
c_1|\nu_1|+
c_2|\nu_2| & \hbox{ otherwise}.
\end{cases}
$$
Note in fact that the normal to any $\partial \{u=1\}$ will be equal to $\nu^k$ ${\cal H}^1$-a.e.~on $ \{ t(\nu^k)^\perp: t\in \R\}+\Z^2$. This shows that $f(y,\cdot)\in {\bf H}_{\alpha,\beta}(\theta)$ for 
${\cal H}^1$-a.a.$y$, and is of the form considered in Proposition \ref{special}.

By a further approximation argument the metrics related to such $f$ can be approximated 
by a sequence 
$$
f^\delta(y,\nu)= \begin{cases}
c^k_1|\nu_1|+c^k_2|\nu_2| & 
\hbox{ if dist} (y, \{ t(\nu^k)^\perp: t\in \R\}+\Z^2)\le\delta\\
& \hbox{ and dist} (y, \{ t(\nu^j)^\perp: t\in \R\}+\Z^2)>\delta,\ j\neq k,  \ k=1,\ldots, N
\\
c_1|\nu_1|+
c_2|\nu_2| & \hbox{ otherwise}.
\end{cases}
$$

By localizing the construction in Proposition \ref{special} we can find $c^{\delta,\eta}_{ij}$ such that
$$
E^\eta(u)={1\over 8}\sum_{ij}\eta c^{\delta,\eta}_{ij} (u_i-u_j)^2 \qquad u:\eta\Z^2\to \{\pm1\}
$$
$\Gamma$-converges as $\eta\to0$ to
$$
\int_{\partial\{u=1\}} f^\delta(x,\nu)d{\cal H}^1
$$
Furthermore, $c^{\delta,\eta}_{ij}$ can be taken periodic of period $1/\eta$ (which we may suppose being integer) and with horizontal and vertical volume fractions $\theta_1$ and $\theta_2$, respectively.

By a diagonal argument this proves the theorem.
\end{proof}

\section{Conclusion and perspectives}
The main purpose of this paper has been the construction of discrete microgeometries, that allow the computation of optimal bounds for mixtures of ferromagnetic interaction. To that end we have dealt with the simplest two-dimensional nearest-neighbour setting. There are several extensions of this results: to higher dimension (where the results will be different for length energies and for discrete perimeter functionals); to energies with long-range interactions (for example for nearest and next-to-nearest interactions, where a multi-scale approach can be necessary); to the computations of the {\em G-closure} of mixtures (i.e., all possible limits of mixtures and not only periodic ones, which, nevertheless, can be reduced to the optimal bounds for periodic mixtures by the {\em localization principle} of Dal Maso and Kohn), to other lattices (e.g., the triangular lattice), etc.

\section*{Acknowledgments}
The author gratefully acknowledge the hospitality of the Mathematical Institute in Oxford and the financial support of the EPSRC Science and Innovation award to the Oxford Centre for Nonlinear PDE (EP/E035027/1).


\begin{thebibliography}{99}%\baselineskip=11pt

\bibitem{AcBu} E. Acerbi and G. Buttazzo. %?
On the limits of periodic Riemannian metrics. 
{\em J. Anal. Math.} {\bf 43} (1983), 183--201. 

\bibitem{ABC}
R. Alicandro, A. Braides, M. Cicalese.
Phase and anti-phase boundaries in binary discrete systems: a variational viewpoint.
{\it Netw. Heterog. Media}  {\bf1} (2006), 85--107

\bibitem{AB}
L. Ambrosio and A. Braides. 
{Functionals defined on partitions of sets of finite perimeter, II: semicontinuity, relaxation and homogenization.}
 {\it J. Math. Pures. Appl.}  {\bf 69} (1990), 307--333. 
 
\bibitem{AI} N. Ansini and O. Iosifescu.
Approximation of anisotropic perimeter functionals by homogenization.
{\it Boll. Un. Mat. Ital.} {\bf 3} (2010), 149--168.
 
 \bibitem{BLBL}  X. Blanc, C. Le Bris and P.L. Lions.
From molecular models to continuum models.
{\it C.R. Acad. Sci., Paris, Ser. I} {\bf  332}  (2001), 949--956.

\bibitem{BBF} A. Braides, G. Buttazzo, and I. Fragal\`a.
Riemannian approximation of Finsler metrics. {\it Asympt. Anal.} {\bf 31} (2002), 177--187

\bibitem{BDF}  {\rm A. Braides and A. Defranceschi}.
{\em  Homogenization of Multiple Integrals}, Oxford University Press,
Oxford, 1998.

\bibitem{BF}A. Braides and G. Francfort,
{Bounds on the effective behavior of a square conducting lattice}. 
{\it R. Soc. Lond. Proc. Ser. A
 Math. Phys. Eng. Sci.} {\bf 460} (2004), 1755--1769 

\bibitem{BP}
A. Braides and A. Piatnitski.
Homogenization of surface and length energies for spin systems.
{\em  J. Funct. Anal.} {\bf 264} (2013), 1296--1328
 
\bibitem{BS}
A. Braides and M. Solci.
Interfacial energies on Penrose lattices.
{\em M3AS} {\bf 21} (2011), 1193--1210


\bibitem{DP} A. Davini and M. Ponsiglione.
Homogenization of two-phase metrics and applications.
{\em J. Analyse Math.} {\bf 103} (2007), 157--196


\end{thebibliography}
\end{document}